\documentclass[10pt,pdftex,letterpaper]{article}
\usepackage[letterpaper,margin=1in]{geometry}

\pdfoutput=1

\newcommand{\TiKzSetup}{
  \usepackage{tikz}
  \usetikzlibrary{shapes}

  \tikzstyle{every node} = [inner sep=0em]
  \tikzstyle{v} = [draw,circle,line width=0.11em,minimum size=0.5em]
  \tikzstyle{cv} = [v,inner sep=0.075em,minimum size=0.9em]
  
  \tikzstyle{e} = [line width=0.15ex]
  \tikzstyle{fe} = [line width=.5ex] 
  \tikzstyle{de} = [dashed,line width=0.125ex] 

  \tikzstyle{treeNode} = [circle,fill=white,inner sep=0.075em]
  \tikzstyle{leaf} = [treeNode,text height=1.5ex,text depth=.25ex]
  \tikzstyle{treeInternal} = [treeNode,draw,font=\small,
                              line width=0.15ex,minimum size=1em]
  \tikzstyle{treeChildRelation} = [line width=0.15ex]
  \tikzstyle{treeDescendantRelation} = [loosely dotted,line width=0.15ex]
}

\newcommand{\twinVertices}[1]{
  \foreach \ang/\lbl/\color/\colnum in {0/a/red!30/3,60/b/green!40/2,
                                        120/c/blue!10/1,180/d/red!30/3,
                                        240/e/blue!10/1,300/f/green!40/2}
    \node[cv,fill=\color,label=\ang:$\lbl$] (\lbl) at (\ang:1){#1\colnum};
}

\newcommand{\figTwinCfive}[3]{
  \begin{tikzpicture}[x=#1/2,y=#1/2,font=#2,]
    \twinVertices{#3}
    \draw[e] (d)--(b)--(a)--(f)--(e)--(d)--(c)--(a);
  \end{tikzpicture}
}

\newcommand{\figCoTwinCfive}[3]{
  \begin{tikzpicture}[x=#1/2,y=#1/2,font=#2,
                     ]
    \twinVertices{#3}
    \draw[e] (d)--(b)--(a)--(f)--(e)--(d)--(c)--(a);
    \draw[e] (c)--(b);
  \end{tikzpicture}
}

\newcommand{\cographVertices}{
  \foreach \x/\y/\ang/\lbl/\color/\colnum in {
           0/3/135/a/red!30/3,
           1/4/135/b/orange!30/4,
           0/1/225/c/purple!40/5,
           1/0/225/d/violet!30/6,
           3/4/45/e/blue!10/1,
           4/3/45/f/green!40/2,
           4/1/-45/g/blue!10/1,
           3/0/-45/h/green!40/2}
}

\newcommand{\cographEdges}{
  \foreach \x in {e,f,g,h}
    \foreach \y in {a,b}
      \foreach \z in {c,d}
        \draw[e] (\x)--(\y)--(\z)--(\x);
  \draw[e] (e)--(f);
  \draw[e] (g)--(h);
}

\newcommand{\coloredCographVertices}[1]{
  \cographVertices
    \node[cv,fill=\color,label=\ang:$\lbl$] (\lbl) at (\x,\y){#1\colnum};
}

\newcommand{\figCographColored}[3]{
  \begin{tikzpicture}[x=#1/5,y=#1/5,font=#2,
                     ]
    \coloredCographVertices{#3}
    \cographEdges
  \end{tikzpicture}
}

\newcommand{\figCographColoredTriangulated}[3]{
  \begin{tikzpicture}[x=#1/5,y=#1/5,font=#2,
                     ]
    \coloredCographVertices{#3}
    \cographEdges
    \draw[fe] (a)--(b);
    \draw[fe] (c)--(d);
  \end{tikzpicture}
}

\newcommand{\figCotreeCanon}[3]{
  \begin{tikzpicture}[x=#1,y=#2,line width=0.15ex,level distance=#2/3,font=#3,
                      level 1/.style={sibling distance=#1/3},
                      level 2/.style={sibling distance=#1/8},
                      level 3/.style={sibling distance=#1/8}]
    \node[treeInternal] {\One} 
      child {node[treeInternal]{\Zero} child {node[leaf]{$a$}}
                                     child {node[leaf]{$b$}}}
      child {node[treeInternal]{\Zero} child[sibling distance=#1/4] {node[treeInternal]{\One} child {node[leaf]{$e$}}
                                                                    child {node[leaf]{$f$}}}
                                     child[sibling distance=#1/4] {node[treeInternal]{\One} child {node[leaf]{$g$}}
                                                                    child {node[leaf]{$h$}}}}
      child {node[treeInternal]{\Zero} child {node[leaf]{$c$}}
                                     child {node[leaf]{$d$}}};
  \end{tikzpicture}
}

\newcommand{\figCotreeBinary}[3]{
  \begin{tikzpicture}[x=#1,y=#2,line width=0.15ex,level distance=#2/3,font=#3,
                      level 1/.style={sibling distance=#1/2},
                      level 2/.style={sibling distance=#1/4},
                      level 3/.style={sibling distance=#1/8}]
    \node[treeInternal] {\One} 
      child {node[treeInternal]{\One} child {node[treeInternal]{\Zero} child {node[leaf]{$a$}}
                                                                    child {node[leaf]{$b$}}}
                                     child {node[treeInternal]{\Zero} child {node[leaf]{$c$}}
                                                                    child {node[leaf]{$d$}}}}
      child {node[treeInternal]{\Zero} child {node[treeInternal]{\One} child {node[leaf]{$e$}}
                                                                    child {node[leaf]{$f$}}}
                                     child {node[treeInternal]{\One} child {node[leaf]{$g$}}
                                                                    child {node[leaf]{$h$}}}};
  \end{tikzpicture}
}

\newcommand{\figCotreeCanonTriangulated}[3]{
  \begin{tikzpicture}[x=#1,y=#2,line width=0.15ex,level distance=#2/3,font=#3,
                      level 1/.style={sibling distance=#1/6},
                      level 2/.style={sibling distance=#1/6},
                      level 3/.style={sibling distance=#1/12}]
    \node[treeInternal]{\One} 
      child {node[leaf]{$a$}}
      child {node[leaf]{$b$}}
      child {node[treeInternal]{\Zero} child {node[treeInternal]{\One} child {node[leaf]{$e$}}
                                                                       child {node[leaf]{$f$}}}
                                       child {node[treeInternal]{\One} child {node[leaf]{$g$}}
                                                                       child {node[leaf]{$h$}}}}
      child {node[leaf]{$c$}}
      child {node[leaf]{$d$}};
  \end{tikzpicture}
}

\newcommand{\figCotreeBinaryTriangulated}[3]{
  \begin{tikzpicture}[x=#1,y=#2,line width=0.15ex,level distance=#2/3,font=#3,
                      level 1/.style={sibling distance=#1/2},
                      level 2/.style={sibling distance=#1/4},
                      level 3/.style={sibling distance=#1/8}]
    \node[treeInternal] {\One} 
      child {node[treeInternal]{\One} child {node[treeInternal]{\One} child {node[leaf]{$a$}}
                                                                    child {node[leaf]{$b$}}}
                                     child {node[treeInternal]{\One} child {node[leaf]{$c$}}
                                                                    child {node[leaf]{$d$}}}}
      child {node[treeInternal]{\Zero} child {node[treeInternal]{\One} child {node[leaf]{$e$}}
                                                                    child {node[leaf]{$f$}}}
                                     child {node[treeInternal]{\One} child {node[leaf]{$g$}}
                                                                    child {node[leaf]{$h$}}}};
  \end{tikzpicture}
}
\newcommand{\figLemmaCotreeCfour}[3]{
  \begin{tikzpicture}[x=#1/3,y=#2/3,font=#3,]
    \node[treeInternal,label=right:{$\tau$}] (tau) at (0,2.5){\One}; 
    \draw[treeDescendantRelation] (tau)--(0,3);

    \node[treeInternal,label=left:$\alpha$] (alpha) at (-1/2,2){};
    \draw[treeChildRelation] (tau)--(alpha);
    \node[treeInternal,label=right:$\beta$] (beta) at (1/2,2){};
    \draw[treeChildRelation] (tau)--(beta);

    \node[treeInternal,label=left:$\alpha'$] (alphaprime) at (-1,1){\Zero};
    \draw[treeDescendantRelation] (alpha)--(alphaprime);
    \node[treeInternal,label=right:$\beta'$] (betaprime) at (1,1){\Zero};
    \draw[treeDescendantRelation] (beta)--(betaprime);

    \foreach \x/\nm/\lbl in {-1.5/a1/$a_1$,-0.5/a2/$a_2$,0.5/b1/$b_1$,1.5/b2/$b_2$}
      \node[leaf](\nm) at (\x,0){\lbl};
    \draw[treeDescendantRelation] (alphaprime)--(a1);
    \draw[treeDescendantRelation] (alphaprime)--(a2);
    \draw[treeDescendantRelation] (betaprime)--(b1);
    \draw[treeDescendantRelation] (betaprime)--(b2);
  \end{tikzpicture}
}

\newcommand{\figLemmaGraphCfour}[2]{
  \begin{tikzpicture}[x=#1/2,y=#1/2,font=#2,
                      text height=1.5ex,text depth=.25ex,]
    \foreach \x/\y/\nm/\lbl/\ang in {-1/1/a1/$a_1$/135,1/1/b1/$b_1$/45,1/-1/a2/$a_2$/-45,-1/-1/b2/$b_2$/-135}
      \node[v,label=\ang:\lbl] (\nm) at (\x,\y){};
    \draw[e] (a1)--(b1)--(a2)--(b2)--(a1);
  \end{tikzpicture}
}

\newcommand{\figLemmaCfour}[3]{
  \begin{tikzpicture}[x=#1,y=#2,font=#3]
    \node at (-1/5,0){\figLemmaGraphCfour{0.1#1}{#3}};
    \node at (1/5,0){\figLemmaCotreeCfour{0.25*#1}{0.175*#2}{#3}};
  \end{tikzpicture}
}

\TiKzSetup
\newcommand{\paperKeywords}{acyclic coloring; star coloring; cograph; treewidth; pathwidth; triangulating colored graphs}
\newcommand{\paperTitle}{Acyclic and Star Colorings of Cographs}
\newcommand{\extendedAbstractTitle}{Acyclic and Star Colorings of Joins of Graphs and an Algorithm for Cographs}
\newcommand{\paperTitleFootnote}{
  A preliminary version of this article, entitled ``\extendedAbstractTitle'', appeared in
  \emph{Proceedings of the 8th Cologne-Twente Workshop on Graphs and Combinatorial Optimization},
  2009~\cite{Lyons09}.
}

\usepackage{amsmath,amsthm,amssymb}
\newtheorem{dfn}{Definition}
\newtheorem{obs}{Observation}

\newtheorem{thm}{Theorem}
\newtheorem{cor}[thm]{Corollary}

\newtheorem{lem}{Lemma}[section]

\usepackage{fancyhdr}

\usepackage{hyperref}
\hypersetup{
  backref=page,%
  pdftitle={\paperTitle},%
  pdfauthor={Andrew Lyons},%
  pdfsubject={},%
  pdfkeywords={\paperKeywords},%
  pdfstartview=FitH,%
  bookmarks=true,%
  bookmarksopen=true,%
  breaklinks=true,%
  colorlinks=true,%
  linkcolor=blue,anchorcolor=blue,citecolor=blue,%
  filecolor=blue,menucolor=blue,urlcolor=blue,
}

\usepackage{amsmath}

\newcommand{\disj}{\cup}

\newcommand{\Zero}{$\disj$}

\newcommand{\join}{\ast} 
\newcommand{\One}{$\join$}

\newcommand{\CotreeNode}[1]{%
  {\tikz[baseline=(X.base)]\node[treeInternal](X){#1};}%
}
\newcommand{\ZeroNode}{
  \CotreeNode{$\disj$}%
}
\newcommand{\OneNode}{
  \CotreeNode{$\join$}%
}

\DeclareMathOperator{\tw}{tw}
\DeclareMathOperator{\pw}{pw}

\newcommand{\card}[1]{\lvert#1\rvert}

\newcommand{\NP}{\mathsf{NP}}

\newcommand{\Gt}{{G^+}}
\newcommand{\Et}{E^+}
\newcommand{\Tt}{{T^+}}

\newcommand{\refdfn}[1]{{Definition~\ref{#1}}}
\newcommand{\refobs}[1]{{Observation~\ref{#1}}}
\newcommand{\refalg}[1]{{Algorithm~\ref{#1}}}
\newcommand{\reffig}[1]{{Figure~\ref{#1}}}
\newcommand{\reffigpart}[2]{{Figure~\ref{#1}(#2)}}

\newcommand{\reflem}[1]{{Lemma~\ref{#1}}}
\newcommand{\refthm}[1]{{Theorem~\ref{#1}}}
\newcommand{\refsec}[1]{{Section~\ref{#1}}}

\newcommand{\refcor}[1]{{Corollary~\ref{#1}}}

\usepackage[boxed,vlined,nofillcomment,procnumbered]{algorithm2e}

\newcommand{\algCOMPUTEAC}{
  \begin{algorithm}[bht]
    \DontPrintSemicolon
    \SetKwInOut{Input}{input}
    \KwData{binary cotree $T$ for a cograph $G$}
    \KwResult{values $\card{V_\tau}$ and $\chi_a(G_\tau)$ for every node $\tau \in T$} 
    \BlankLine
    \SetKwFunction{COMPUTEAC}{COMPUTEAC}
    \ProcSty{Procedure} \COMPUTEAC{$t$}\;
    \Input{cotree node $\tau$ from $T$}
    \Begin{
      \eIf{$\tau$ is a leaf}{
        $\card{V_\tau} \longleftarrow 1$\;
        $\chi_a(G_\tau) \longleftarrow 1$\;
      }( \tcc*[h]{$\tau$ has children $\alpha$ and $\beta$}){
        \COMPUTEAC{$\alpha$}\;
        \COMPUTEAC{$\beta$}\;
        $\card{V_\tau} = \card{V_\alpha}+\card{V_\beta}$\;
        \eIf{$t$ is a \ZeroNode-node}{
          $\chi_a(G_\tau) \longleftarrow \max\bigl\{\chi_a(G_\alpha),\chi_a(G_\beta)\bigr\}$\;
        }( \tcc*[h]{$\tau$ is a \OneNode-node}){
          $\chi_a(G_\tau) \longleftarrow
              \min\bigl\{ \card{V_\alpha}+\chi_a(G_\beta), \card{V_\beta}+\chi_a(G_\alpha) \bigr\}
          $\;
        }
      }
    }
    \caption{Calling \protect\COMPUTEAC{$\rho$} where $\rho$ is the root of $T$
             yields $\chi_a(G)$ in $O(n)$ time.}
    \label{alg:phaseI}
  \end{algorithm}
}

\newcommand{\procSATURATE}{
  \SetKwFunction{SATURATE}{SATURATE}
  \ProcSty{Procedure} \SATURATE{$\tau$, $K$}\;
  \Input{node $\tau \in T$, positive integer $K$}
  \Begin{
    \eIf{$\tau$ is a leaf}{
      $\phi(\tau) \longleftarrow K$\;
    }( \tcc*[h]{$\tau$ has children $\alpha$ and $\beta$}){
      \SATURATE{$\alpha$, $K$}\; 
      \SATURATE{$\beta$, $K+\card{V_\alpha}$}\;
    }
  }
}

\newcommand{\procASSIGNCOLORS}{
  \SetKwFunction{ASSIGNCOLORS}{ASSIGNCOLORS}
  \ProcSty{Procedure} \ASSIGNCOLORS{$\tau$, $K$}\;
  \Input{node $\tau$ from $T$, positive integer $K$}
  \Begin{
    \eIf{$\tau$ is a leaf}{
      $\phi(\tau) \longleftarrow K$\;
    }( \tcc*[h]{$\tau$ has children $\alpha$ and $\beta$}){
      \eIf{$\tau$ is a \ZeroNode-node}{
        \ASSIGNCOLORS{$\alpha$, $K$}\;
        \ASSIGNCOLORS{$\beta$, $K$}\;
      }( \tcc*[h]{$\tau$ is a \OneNode-node}){
        \eIf{$\card{V_\alpha}+\chi_a(G_\beta) \leq \card{V_\beta}+\chi_a(G_\alpha)$}{
          \SATURATE{$\alpha$, $K$}\;
          \ASSIGNCOLORS{$\beta$, $K+\card{V_\alpha}$}\;
        }( \tcc*[h]{$\card{V_\beta}+\chi_a(G_\alpha) < \card{V_\alpha}+\chi_a(G_\beta)$}){
          \SATURATE{$\beta$, $K$}\;
          \ASSIGNCOLORS{$\alpha$, $K+\card{V_\beta}$}\;
        }
      }
    }
  }
}

\newcommand{\algASSIGNCOLORS}{
  \begin{algorithm}[bht]
    \DontPrintSemicolon
    \SetKwInOut{Input}{input}
    \KwData{binary cotree $T$ for a cograph $G$;
            values $\card{V_\tau}$ and $\chi_a(G_\tau)$ for every node $\tau$ in $T$}
    \KwResult{acyclic coloring $\phi\colon V \rightarrow \{1,\ldots,\chi_a(G)\}$}
    \BlankLine
    \procSATURATE
    \BlankLine
    \procASSIGNCOLORS
    \caption{calling \protect\ASSIGNCOLORS{$\rho$, $1$} where $\rho$ is the root of $T$
             yields an optimal acyclic coloring $G$ in $O(n)$ time.
             The obtained coloring is also an optimal star coloring of $G$.}
    \label{alg:phaseII}
  \end{algorithm}
}

\title{ \Large \bf \paperTitle\footnote{\paperTitleFootnote}}
\author{\normalsize Andrew Lyons \\[2mm]
        \small Department of Computer Science, Dartmouth College \\
        \small \href{mailto:lyonsam@gmail.com}{\texttt{lyonsam@gmail.com}}
}
\date{}
\begin{document}
\maketitle
\pagestyle{plain}
\thispagestyle{plain}

\begin{abstract}
  An \emph{acyclic coloring} of a graph is a proper vertex coloring such that
the union of any two color classes induces a disjoint collection of trees.
The more restricted notion of \emph{star coloring} requires that
the union of any two color classes induces a disjoint collection of stars.
We prove that every acyclic coloring of a cograph is also a star coloring
and give a linear-time algorithm for finding an optimal acyclic and star coloring of a cograph.
If the graph is given in the form of a cotree, the algorithm runs in $O(n)$ time.
We also show that the acyclic chromatic number, the star chromatic number,
the treewidth plus one, and the pathwidth plus one are all equal for cographs.

\end{abstract}

\section{Introduction}\label{sec:intro}
A {\em proper vertex coloring} (or {\em proper coloring}) of a graph $G$
is a mapping $\phi : V \rightarrow \mathbb{N}^+$ 
such that if $a$ and $b$ are adjacent vertices, then $\phi(a) \neq \phi(b)$.
The chromatic number of a graph $G,$ denoted $\chi(G),$
is the minimum number of colors required in any proper coloring of $G.$
An {\em acyclic coloring} of a graph is a proper coloring such that the subgraph induced by
the union of any two color classes is a disjoint collection of trees.
A {\em star coloring} of a graph is a proper coloring such that the subgraph induced by
the union of any two color classes is a disjoint collection of stars.
The acyclic and star chromatic numbers of $G$ are defined analogously to the chromatic number
and are denoted by $\chi_a(G)$ and $\chi_s(G),$ respectively.
Since a disjoint collection of stars constitutes a forest,
it follows that every star coloring is also an acyclic coloring and $\chi_a(G) \leq \chi_s(G)$ for every graph $G.$
We will often find it useful to work with the alternative definitions
implied by the following (folklore) observation.
\begin{obs}\label{obs:acystarchar}
  The following are true whenever $\phi$ is a proper coloring of a graph $G$.
  \begin{itemize}
    \item[] $\phi$ is an acyclic coloring of $G$ if and only if every cycle in $G$ uses at least three colors.
    \item[] $\phi$ is a star coloring of $G$ if and only if every path on four vertices in $G$ uses at least three colors.
  \end{itemize}
\end{obs}

A great deal of graph-theoretical research has been conducted on acyclic and star coloring
since they were introduced in the early seventies by Gr\"unbaum~\cite{Grünbaum1973}.
Our investigation of these problems from an algorithmic point of view is motivated in part by their
applications in combinatorial scientific computing,
where they model the optimal evaluation of sparse Hessian matrices.
In fact, these coloring problems were independently discovered and studied by the scientific computing community.
The survey of Gebremedhin et al.~\cite{Gebremedhin2005WCI} gives a history of the subject
as well as an overview of the use of these coloring variants in computing sparse derivative matrices.

The acyclic and star coloring problems are both $\NP$-hard,
and most results concerning their complexity on special classes of graphs are negative.
In particular, both problems remain $\NP$-hard even when restricted to bipartite graphs~\cite{CC1986CCP,CM1984ESH}.
In addition, Albertson et al.~\cite{AlbertsonCKKR04} showed that the problem of determining
whether the star chromatic number is at most three is $\NP$-complete even for planar bipartite graphs.
The authors also showed that it is $\NP$-complete to decide whether the chromatic number of a graph $G$
is equal to the star chromatic number of $G$, even if $G$ is a planar graph with chromatic number three.
Inapproximability results for both problems are given in~\cite{GTMP2007}.

Researchers have obtained a few positive algorithmic results for these problems on graphs for which
the acyclic or star chromatic number is bounded by a constant.
In particular, Skulrattanakulchai~\cite{Skulrattanakulchai2004161} gives a linear-time algorithm for
finding an acyclic coloring of a graph with maximum degree three that uses four colors or fewer,
and Fertin and Raspaud~\cite{Fertin200865} give a linear-time algorithm for
finding an acyclic coloring of a graph with maximum degree five that uses nine colors or fewer.
To our knowledge, prior to this work no polynomial time algorithm was known for either of these problems
on a nontrivial class of graphs for which the acyclic or star chromatic number is unbounded. 

In this paper, we consider acyclic and star colorings of cographs.
This class of graphs, which we define formally in \refsec{sec:cographs},
was discovered independently by a number of researchers,
and hence has many characterizations.
We refer the interested reader to the book of
Brandst{\"a}dt, Le, and Spinrad~\cite{BLS1999} for a additional background and
details related to cographs.
Many problems that are $\NP$-complete on general graphs have polynomial time algorithms when restricted to cographs,
in part because of the nice decomposition properties that these graphs exhibit.
Nevertheless, problems such as list coloring and achromatic number remain $\NP$-complete on this class~\cite{Bodlaender89,JS1997}.
Our motivation, however, stems also from a new characterization of cographs:
they are exactly the graphs for which every acyclic coloring is also a star coloring.
We begin \refsec{sec:cographs} with a simple proof of this fact.

Bodlaender and M\"{o}hring~\cite{BM93} showed that the pathwidth of a cograph equals its treewidth.
In \refsec{sec:triangulations}, we prove that the acyclic colorings of a cograph $G$ coincide with
the proper colorings of triangulations of $G.$
As a consequence, we find that the acyclic chromatic number, the star chromatic number,
the treewidth plus one, and the pathwidth plus one are all equal for cographs.
Additionally, we discuss some implications of our results for
the triangulating colored graphs problem,
which is related to a problem from evolutionary biology.

In \refsec{sec:algorithm}, we describe an algorithm that, given a cograph $G$,
produces an optimal acyclic and star coloring of $G$.
When $G$ is given as an adjacency list, our algorithm runs in $O(n+m)$ time,
where $n$ and $m$ are the numbers of vertices and edges in $G$, respectively;
only $O(n)$ time is required when $G$ is given in the form of a cotree,
which is a concise tree-based representation of $G$.

\section{Cographs and cotrees}\label{sec:cographs}
Let $G_1 = (V_1,E_1)$ and $G_2 = (V_2,E_2)$ be graphs such that $V_1 \cap V_2 = \emptyset$.
The {\em disjoint union} of $G_1$ and $G_2$ is the graph $G_1 \disj G_2 = (V_1 \cup V_2, E_1 \cup E_2)$.
The {\em join} of $G_1$ and $G_2$, denoted $G_1 \join G_2$,
 is the graph obtained by adding all possible edges between $G_1$ and $G_2$, i.e.,
$G_1 \join G_2 = \left(V_1 \cup V_2, E_1 \cup E_2 \cup \{v_1v_2 \mid v_1 \in V_1, v_2 \in V_2 \}\right)$.
\begin{dfn}[cograph~\cite{CPS1985}]\label{dfn:cograph}
  A graph $G = (V,E)$ is a \emph{cograph} if and only if one of the following conditions holds:
  \begin{enumerate}
    \item $\card{V} = 1;$
    \item there exist cographs $G_1,\ldots,G_k$ such that $G = G_1 \disj G_2 \disj \cdots \disj G_k;$
    \item there exist cographs $G_1,\ldots,G_k$ such that $G = G_1 \join G_2 \join \cdots \join G_k.$
  \end{enumerate}
\end{dfn}
As noted in \refsec{sec:intro}, this class can be characterized in a number of ways.
We will find the following characterization most useful.
\begin{thm}[{\cite[Theorem 11.3.3]{BLS1999}}]\label{thm:cographchar}
  The cographs are exactly those graphs that
  do not contain an induced path on four vertices.
\end{thm}

\begin{thm} \label{thm:StarAcyclic}
  The cographs are exactly those graphs $G$ for which
  every acyclic coloring of $G$ is also a star coloring of $G$.
\end{thm}
\begin{proof}
  ($\Rightarrow$):
  Let $\phi$ be an acyclic coloring of a cograph $G$, and
  let $P$ be a path on four vertices in $G$.
  Since, by \refthm{thm:cographchar}, $G$ cannot contain an induced path on four vertices,
  the graph induced by $P$ must either
  be a cycle (in which case it uses at least three colors by \refobs{obs:acystarchar})
  or contain a triangle
  (in which case it uses at least three colors because $\phi$ is a proper coloring).
  Thus, by \refobs{obs:acystarchar}, $\phi$ is a star coloring of $G$.

  ($\Rightarrow$):
  Now suppose $G$ is a graph in which every acyclic coloring is also a star coloring
  and assume for the sake of contradiction that $G$ contains
  an induced path on vertices $abcd$ (in that order).
  Let $\phi$ be a coloring that assigns each vertex in $G$ its own (distinct) color
  with the exceptions $\phi(a)=\phi(c)$ and $\phi(b)=\phi(d)$.
  Because $abcd$ induces a path, we have that $\phi$ is both proper and acyclic,
  yet, by \refobs{obs:acystarchar}, $\phi$ is not a star coloring,
  which is a contradiction.
\end{proof}
\begin{cor} \label{cor:StarAcyclic}
  For every cograph $G$, $\chi_s(G) = \chi_a(G).$
\end{cor}
The fact that every acyclic coloring of a cograph is also a star coloring means that,
for the bulk of this paper, we may restrict our attention to acyclic colorings.

\subsection{Cotrees}
Cographs can be recognized in linear time~\cite{CPS1985,HP2005},
and most recognition algorithms also produce a special decomposition structure
in the same time bound when the input graph $G$ is a cograph.
We now introduce this structure, which is often used in algorithms designed to work on cographs.
A \emph{cotree} for a cograph $G$ is a rooted tree $T$
whose leaves correspond to the vertices of $G$ and whose internal nodes
are given labels from $\{\join,\disj\}$ such that
two vertices in $G$ are adjacent if and only if the lowest common ancestor
of the corresponding leaves in $T$ is a \OneNode-node.
(For the sake of clarity, we will use the word ``node'' when referring to cotrees,
whereas the term ``vertex'' will be reserved for the context of the original graph $G$.)
For a node $\tau$ in $T$, $V_\tau$ denotes the set of vertices in $G$ that
correspond to leaves in the subtree of $T$ rooted at $\tau$;
we denote by $G_\tau$ the subgraph of $G$ induced by $V_\tau$.
Cotrees thus describe the recursive construction of cographs
described in \refdfn{dfn:cograph}.

While there may be many cotrees for a given cograph $G$,
the \emph{canonical cotree} of $G$, which is characterized by the property that
any path from a leaf to the root alternates between \ZeroNode-nodes and \OneNode-nodes,
is unique up to isomorphism~\cite{CLS1981}. 
It is often more convenient to work with cotrees whose
internal nodes have exactly two children.
Since the operations $\disj$ and $\join$ are commutative and associative,
one can show that any cotree $T$ can, in linear time,
be converted into a cotree $T'$ such that $T'$ that meets this condition
and has size linear in that of $T$~\cite{BM93}.
An example is shown in \reffig{fig:cographAndCotrees}.
\begin{figure}[bht]
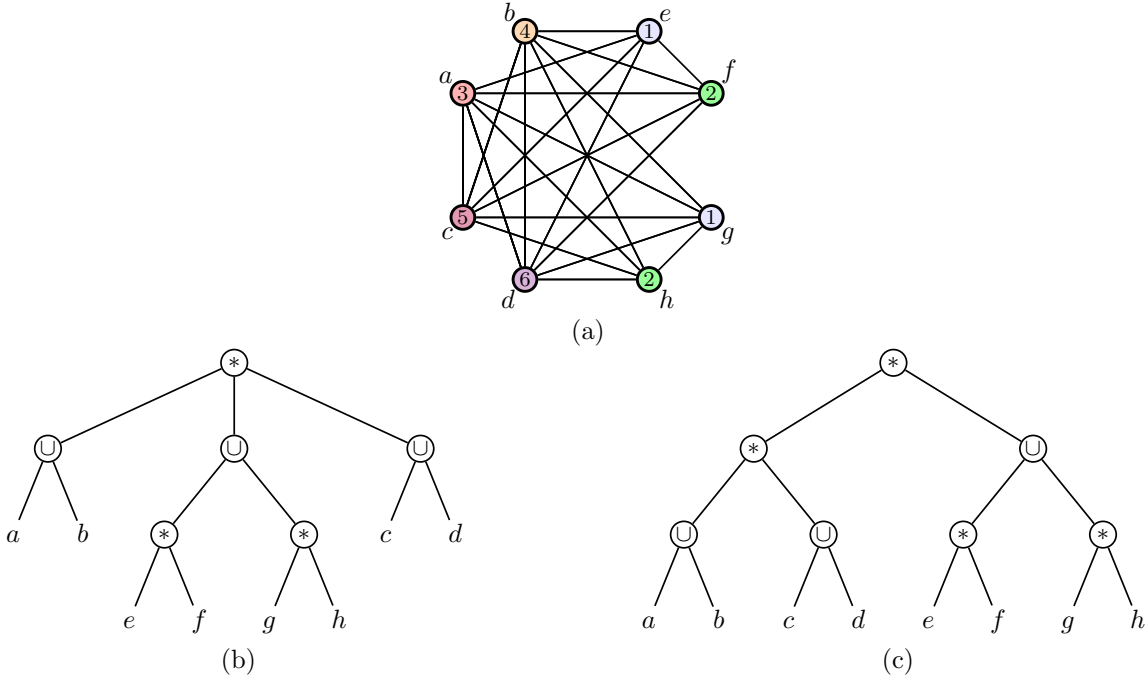
 \centering
  \figCographColored{0.25\textwidth}{\normalsize}{\footnotesize} \\
  (a) \\
  \begin{tabular*}{0.95\textwidth}{@{\extracolsep{\fill}}cc}
    \figCotreeCanon{0.45\textwidth}{0.15\textheight}{\normalsize} &
    \figCotreeBinary{0.45\textwidth}{0.15\textheight}{\normalsize} \\
    (b) & (c)
  \end{tabular*}
  \caption{(a) A cograph $G$; (b) its canonical cotree; (c) a binary cotree.
           The graph $G$ is shown along with an optimal acyclic coloring, which is
           (necessarily) also an optimal star coloring by \refthm{thm:StarAcyclic}.}
  \label{fig:cographAndCotrees}
\end{figure}

We conclude this section with a lemma that suggests a natural way that
a binary cotree can be used to compute
the acyclic and star chromatic numbers of a cograph;
a formal description and analysis of our algorithm is given in \refsec{sec:algorithm}.
Note that this result applies to disjoint unions and joins of \emph{general} graphs,
and is analogous to a result of Bodlaender and M{\"o}hring~\cite[Lemma 3.4]{BM93}
concerning treewidth and pathwidth (we define these notions in the next section).
\begin{lem}\label{lem:acyStarJoinDisj}
  The following hold for any graphs $G_1 = (V_1,E_1)$ and $G_2 = (V_2,E_2).$
  \begin{enumerate}
    \item $\chi_a(G_1 \disj G_2) = \max\bigl\{\chi_a(G_1),\chi_a(G_2)\bigr\};$
    \item $\chi_s(G_1 \disj G_2) = \max\bigl\{\chi_s(G_1),\chi_s(G_2)\bigr\};$
    \item $\chi_a(G_1 \join G_2) = \min\bigl\{\chi_a(G_1)+\card{V_2},\chi_a(G_2)+\card{V_1}\bigr\};$
    \item $\chi_s(G_1 \join G_2) = \min\bigl\{\chi_s(G_1)+\card{V_2},\chi_s(G_2)+\card{V_1}\bigr\}.$
  \end{enumerate}
\end{lem}
\begin{proof}
  The proofs of (i) and (ii) are obvious.

  Our proof of (iii) begins by showing that $\chi_a(G_1 \join G_2) \leq \min\bigl\{\chi_a(G_1)+\card{V_2},\chi_a(G_2)+\card{V_1}\bigr\}$.
  We describe an algorithm that, given optimal acyclic colorings of $G_1$ and $G_2$,
  produces an acyclic coloring $\phi$ of $G_1 \join G_2$ that uses the desired number of colors.
  Let $\phi_1$ and $\phi_2$ be arbitrary optimal acyclic colorings of $G_1$ and $G_2,$ respectively.
  Assume without loss of generality that $\chi_a(G_2)+\card{V_1} \leq \chi_a(G_1)+\card{V_2}.$
  We construct $\phi$ as follows.
  Color those vertices in $V_2$ the same as they are colored by $\phi_2$, using the colors in $\{1,\ldots,\chi_a(G_2)\}$.
  Color those vertices in $V_1$ such that each $v \in V_1$ receives a distinct color in $\{\chi_a(G_2)+1,\ldots,\chi_a(G_2)+\card{V_1}\}$.
  Suppose that $\phi$ causes a bichromatic cycle $C \subseteq V$ in $G_1 \join G_2$.
  Since each vertex in $V_1$ gets a distinct color and no vertex in $V_1$ shares a color with a vertex in $V_2$,
  it follows that any bicromatic cycle in $G_1 \join G_2$ must be contained entirely in $V_2$,
  which contradicts the fact that $\phi_2$ is an acyclic coloring of
  the subgraph of $G_1 \join G_2$ induced by $V_2.$
  Thus $\phi$ is an acyclic coloring of $G_1 \join G_2$, which completes this direction of the proof.

  Now let $\phi$ be an optimal acyclic coloring of $G_1 \join G_2$.
  Observe that, since every vertex from $V_1$ is adjacent to every vertex from $V_2$,
  $V_1$ and $V_2$ must receive disjoint sets of colors from $\phi$.
  Moreover, $\phi$ must assign $\card{V_i}$ distinct colors to the vertices in $V_i$
  for some $i \in \{1,2\}$, as otherwise there would be
  vertices $a_1,b_1 \in V_1$ such that $\phi(a_1) = \phi(b_1)$ and
  vertices $a_2,b_2 \in V_2$ such that $\phi(a_2) = \phi(b_2)$,
  and thus $\phi$ would color the cycle $a_1a_2b_1b_2$ with only two colors
  (contradicting the fact that $\phi$ is an acyclic coloring).
  Finally, since
  $\phi$ must be an acyclic coloring of every induced subgraph of $G_1 \join G_2$,
  we have that $\phi$ must assign at least $\chi_a(G_i)$ colors to
  the subgraph induced by $V_i$ for $i \in \{1,2\}$,
  which implies that $\phi$ must use at least
  $\min\bigl\{\chi_a(G_1)+\card{V_2},\chi_a(G_2)+\card{V_1}\bigr\}$ colors,
  which completes the proof of (iii).

  The proof of (iv) is similar to that of (iii).
\end{proof}

\section{Acyclic colorings and triangulations of cographs}\label{sec:triangulations}

A graph is \emph{chordal} if it has no induced cycle on four or more vertices, i.e.,
every cycle of length greater than three has a \emph{chord}.
A \emph{triangulation} of a graph $G = (V,E)$ is
a chordal graph $\Gt = (V,\Et)$ such that $E \subseteq \Et$.
The following notion allows us to characterize those cographs which are also chordal
in terms of the structure of their cotrees.
\begin{dfn}[skew cotree]
  A cotree $T$ is said to be \emph{skew} if
  at most one child of every \OneNode-node in $T$ has some \ZeroNode-node as a descendant.
\end{dfn}
\begin{lem}\label{lem:cyclesincotrees}
  A cograph $G$ contains an induced cycle on four or more vertices if and only if
  $T$ is not skew for every cotree $T$ of $G$.
\end{lem}
\begin{proof}
  ($\Rightarrow$):
  Suppose $G$ contains an induced cycle on four or more vertices.
  Since $G$ is a cograph, and thus cannot contain an induced path on four vertices,
  such a cycle must consist of \emph{exactly} four vertices,
  say $a_1b_1a_2b_2$ (in that order).
  Let $T$ be an arbitrary cotree of $G$ and denote by
  $\tau$ the lowest common ancestor in $T$ of
  the leaves corresponding to $a_1$ and $b_1$.
  Since $a_1$ and $b_1$ are adjacent, we have that $\tau$ is a \OneNode-node and
  has distinct children $\alpha$ and $\beta$ such that
  $\alpha$ is an ancestor of $a_1$ and
  $\beta$ is an ancestor of $b_1$.
  We will demonstrate that $T$ is not skew by showing that 
  both $\alpha$ and $\beta$ are ancestors of \ZeroNode-nodes.
  Let $\alpha'$ be the lowest common ancestor in $T$ of $a_1$ and $a_2$.
  Since $\alpha'$ is a \ZeroNode-node (because $a_1$ and $a_2$ are non-adjacent),
  $\alpha'$ must be distinct from $\tau$.
  At the same time, $\alpha'$, $\alpha$, and $\tau$ are all ancestors of $a_1$,
  implying that all of these nodes lie on the unique path in $T$ from $a_1$ to the root.
  It follows that $\alpha'$ must be a descendant of $\tau$, as otherwise
  the lowest common ancestor of $a_2$ and $b_1$ would be $\alpha'$ (a \ZeroNode-node).
  In particular, $\alpha'$ must also be a descendant of $\alpha$ because
  $\alpha$ is a child of $\tau$ (note that we may have $\alpha' = \alpha$).
  Similar reasoning with respect to the lowest common ancestor $\beta'$ in $T$ of $b_1$ and $b_2$
  implies that $\beta'$ is a \ZeroNode-node and a descendant of $\beta$.
  It follows that $T$ is not skew, which completes this direction of the proof.

  \begin{figure}[bht]
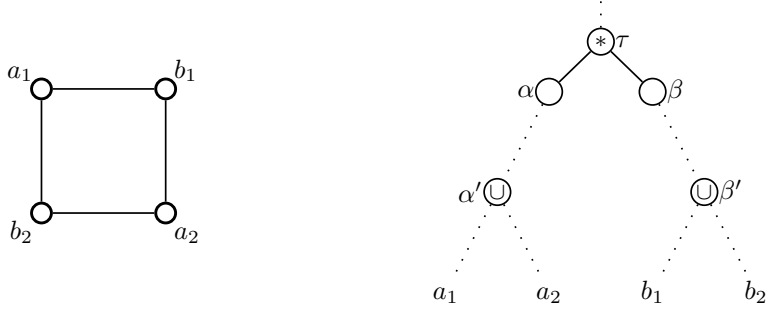
 \centering
    \figLemmaCfour{\textwidth}{\textheight}{\normalsize}
    \caption{The cotree structure of an induced cycle on four vertices in a cograph,
             as in the proof of \reflem{lem:cyclesincotrees}.}
    \label{fig:lemcyclesincotrees}
  \end{figure}

  ($\Leftarrow$):
  Now suppose there exists some cotree $T$ of $G$ that is not skew.
  It follows that $T$ contains a \OneNode-node $\tau$ such that
  $\tau$ is the lowest common ancestor in $T$ of two distinct \ZeroNode-nodes
  which we will call $\alpha'$ and $\beta'$.
  Now let $a_1$ and $a_2$ (resp. $b_1$ and $b_2$) be any pair of leaves whose
  lowest common ancestor in $T$ is $\alpha'$ (resp. $\beta'$).
  Since the lowest common ancestor in $T$ of $a_i$ and $b_j$ is
  $\tau$ for $i,j \in \{1,2\}$, and $\tau$ is a \OneNode-node,
  we have that $a_1b_1a_2b_2$ induces a cycle in $G$,
  which completes the proof.
\end{proof}

\begin{dfn}
  Let $\phi$ be a proper coloring of a cograph $G$ and let $T$ be a cotree of $G$.
  A node $\tau$ in $T$ is said to be \emph{saturated by $\phi$} if
  $\phi(x) \neq \phi(y)$ for all distinct $x,y \in V_\tau$.
\end{dfn}
\begin{lem}\label{lem:acychar}
  If $\phi$ is a proper coloring of a cograph $G$, then the following are equivalent.
  \begin{enumerate}
    \item $\phi$ is an acyclic coloring of $G$;
    \item if $T$ is a cotree of $G$, then
          for every \OneNode-node $\tau$ in $T$,
          at most one child of $\tau$ is not saturated by $\phi$;
    \item $\phi$ is a proper coloring of some triangulation of $G$.
  \end{enumerate}
\end{lem}
\begin{proof}
  ((i)$\Rightarrow$(ii)):
  Suppose $\phi$ is an acyclic coloring of $G$ and
  assume for the sake of contradiction that there exists some \OneNode-node $\tau$ in $T$
  with distinct children $\alpha$ and $\beta$ such that
  neither $\alpha$ nor $\beta$ is $\phi$-saturated.
  It follows that $\alpha$ is an ancestor of leaves $a_1$ and $a_2$ such that
  $\phi(a_1) = \phi(a_2)$ and,
  likewise, $\beta$ is an ancestor of distinct leaves $b_1$ and $b_2$
  such that $\phi(b_1) = \phi(b_2)$.
  It follows that, with respect to $\phi$, $a_1b_1a_2b_2$ induces a bichromatic cycle in $G$,
  which contradicts the fact that $\phi$ is an acyclic coloring of $G$.

  ((ii)$\Rightarrow$(iii)):
  Now suppose $\phi$ satisfies (ii) and let $T$ be an arbitrary cotree of $G$.
  The desired triangulation of $G$ is constructed as follows.
  For every node $\tau$ in $T$ that is saturated by $\phi$,
  turn the subgraph $G_\tau$ into a clique by changing
  every \ZeroNode-node in the subtree of $T$ rooted at $\tau$ into a \OneNode-node.
  By (ii), the resulting cotree $\Tt$ is skew, and
  \reflem{lem:cyclesincotrees} implies that $\Tt$ is
  a cotree of a triangulation $\Gt$ of $G$.
  Since $\phi$ is a proper coloring of $G$, and
  none of the new edges added to form $\Gt$ connect vertices of the same color,
  it follows that $\phi$ is a proper coloring of $\Gt$ as desired.
  
  ((iii)$\Rightarrow$(i)):
  Suppose $\phi$ is a proper coloring of some triangulation $\Gt$ of $G$.
  Since every cycle in $G$ is also a cycle in $\Gt$ (which is a supergraph of $G$), and
  every cycle in $\Gt$ has a chord, it follows that
  $\phi$ must use at least three colors for every cycle in $G$,
  thus $\phi$ is an acyclic coloring of $G$.
\end{proof}

\subsection{Treewidth and pathwidth of cographs}
The \emph{clique number} of a graph $G$, denoted $\omega(G)$, is
the largest number of pairwise adjacent vertices in $G.$
The \emph{treewidth} of a graph $G,$ denoted $\tw(G)$, is
the minimum value of $\omega(\Gt)-1$ over all triangulations $\Gt$ of $G$.
\begin{thm}[folklore]\label{thm:acyTW}
  For every graph $G$, $\chi_a(G) \leq \tw(G)+1$.
\end{thm}
\begin{proof}
  By the definition of treewidth, there exists a triangulation $\Gt$ of $G$ such that
  $\omega(\Gt) = \tw(G)+1$.
  Since $\Gt$ is chordal and since chordal graphs are perfect~\cite{Golumbic2004},
  $\Gt$ further satisfies $\omega(\Gt) = \chi(\Gt)$.
  The desired inequality then follows from the observation that
  every proper coloring of $\Gt$ is an acyclic coloring of $G$
  (as shown in the proof of \reflem{lem:acychar}).
\end{proof}
\begin{thm}\label{thm:AcyclicTree}
  For every cograph $G$, $\chi_a(G) = \tw(G)+1.$
\end{thm}
\begin{proof}
  By \refthm{thm:acyTW}, it suffices to show that $\chi_a(G) \geq \tw(G)+1$.
  To see that this is so, consider an arbitrary optimal acyclic coloring of $G$ which,
  by \reflem{lem:acychar}, is a proper coloring of some triangulation $\Gt$ of $G$.
  The desired inequality then follows from the fact that
  $\chi_a(G) \geq \omega(\Gt) \geq \tw(G)+1$.
\end{proof}

A graph is an \emph{interval graph} if its vertices can be put in correspondence with intervals on the real line
such that two vertices are adjacent if and only if the corresponding intervals have a nonempty intersection.
An \emph{intervalization} of a graph $G = (V,E)$ is an interval graph $\Gt = (V,\Et)$ such that $E \subseteq \Et$.
The \emph{pathwidth} of a graph $G,$ denoted $\pw(G)$, is the minimum value of $\omega(\Gt)-1$ over all intervalizations $\Gt$ of $G$.
Note that since the interval graphs form a proper subclass of the chordal graphs,
we have that $\tw(G) \leq \pw(G)$ for all graphs $G$.
Bodlaender and M\"{o}hring obtained the following result by showing that
every triangulation of a cograph $G$ is also an intervalization of $G.$
\begin{thm}[\cite{BM93}]\label{thm:TreePath}
  For every cograph $G$, $\tw(G) = \pw(G)$.
\end{thm}

Combining \refcor{cor:StarAcyclic}, \refthm{thm:AcyclicTree}, and \refthm{thm:TreePath} we obtain the following result.
\begin{cor}\label{cor:AcyStarTwPw}
  For every cograph $G$, $\chi_s(G) = \chi_a(G) = \tw(G)+1 = \pw(G)+1$.
\end{cor}

We note that \refcor{cor:AcyStarTwPw} also follows from
an inductive application of \reflem{lem:acyStarJoinDisj}.
However, the results of the next section,
as well as the proof of correctness of the algorithm given in \refsec{sec:algorithm},
rely on the intermediate results used in the proof we have given here.

\subsection{Triangulations of colored graphs and the perfect phylogeny problem}
Let $G$ be a graph given with a proper coloring $\phi$.
We say that $G$ is \emph{$\phi$-triangulatable} if there exists
a triangulation $\Gt$ of $G$ such that $\phi$ is a proper coloring of $\Gt$.
In the general case, determining whether $G$ is $\phi$-triangulatable
is $\NP$-complete~\cite{BFHWW2000}.
This is known as the \emph{triangulating colored graphs} problem,
which is polynomially equivalent to
the perfect phylogeny problem from evolutionary biology~\cite{KW1992}.
The following result follows immediately from \reflem{lem:acychar}, which
characterizes the colorings $\phi$ for which a cograph $G$ is $\phi$-triangulatable as
exactly the acyclic colorings of $G$.
(Note that we can check in polynomial time whether $\phi$ is an acyclic coloring of $G$,
and the procedure described in the proof of \reflem{lem:acychar} can be used to
obtain a compatible triangulation in case one is desired.)
\begin{cor}
  There exists a polynomial-time algorithm that,
  given a cograph $G$ along with a proper coloring $\phi$ of $G$,
  determines whether $G$ is $\phi$-triangulatable.
\end{cor}
\reffig{fig:cographAndCotreesTriangulated} illustrates this concept for
the graph depicted in \reffigpart{fig:cographAndCotrees}{a}.
\begin{figure}[bht]
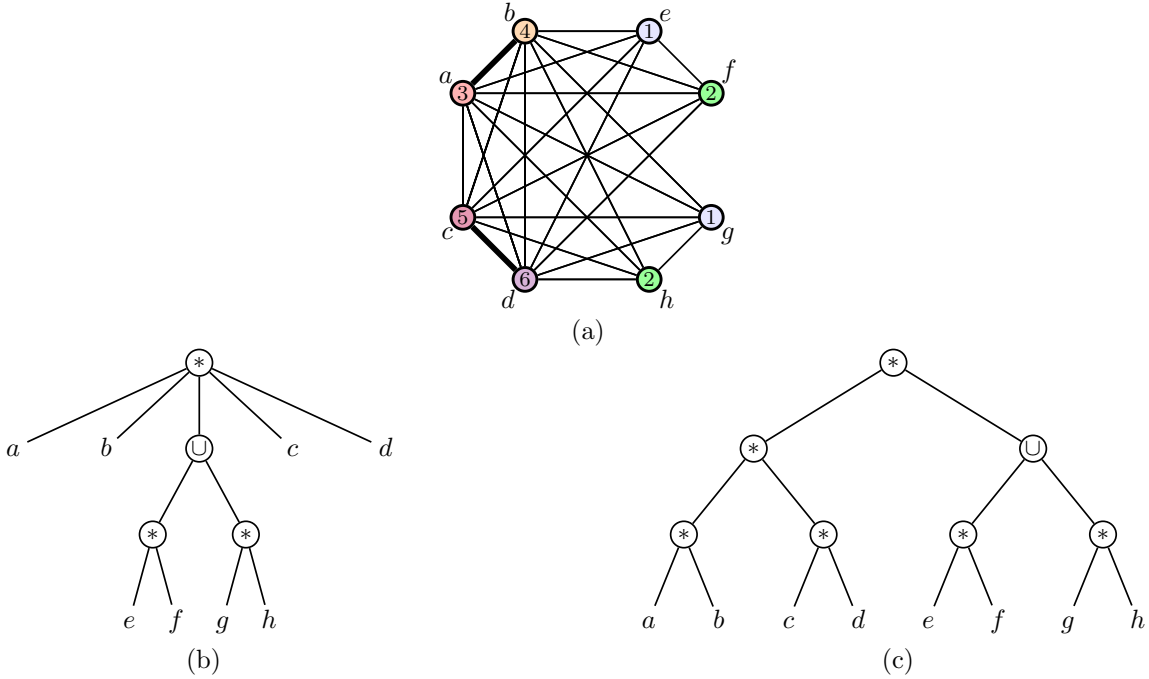
 \centering
  \figCographColoredTriangulated{0.25\textwidth}{\normalsize}{\footnotesize} \\
  (a) \\
  \begin{tabular*}{0.95\textwidth}{@{\extracolsep{\fill}}cc}
    \figCotreeCanonTriangulated{0.45\textwidth}{0.15\textheight}{\normalsize} &
    \figCotreeBinaryTriangulated{0.45\textwidth}{0.15\textheight}{\normalsize} \\
    (b) & (c)
  \end{tabular*}
  \caption{(a) A triangulation $\Gt$ of the graph $G$ depicted in \reffigpart{fig:cographAndCotrees}{a}
               (bold edges added during the triangulation process), which
               satisfies $\chi_a(G) = \chi(\Gt) = \omega(\Gt) = 6;$
           (b) the canonical cotree for $\Gt$;
           (c) the binary cotree $\Tt$ for $\Gt$ obtained from
               the binary cotree depicted in \reffigpart{fig:cographAndCotrees}{c}.}
  \label{fig:cographAndCotreesTriangulated}
\end{figure}

\section{The algorithm for acyclic and star coloring}\label{sec:algorithm}
In this section we prove the following theorem.
\begin{thm}\label{thm:algWorks}
  There exists an algorithm that, given a binary cotree $T$ for cograph $G$,
  produces an optimal acyclic and star coloring of $G$ in $O(n)$ time.
  Moreover, the obtained coloring is also an optimal star coloring.
\end{thm}
As mentioned in \refsec{sec:cographs},
an arbitrary cotree can be transformed into a binary cotree in $O(n)$ time.
If $G$ is given in the form of an adjacency list,
we simply build a cotree of $G$ (in $O(n+m)$ time~\cite{CPS1985})
before running the algorithm,
which results in an overall running time of $O(n+m)$.
We now give an informal description of the algorithm, which consists of two phases.

In the first phase, we traverse the cotree from the leaves to the root,
computing for every $\tau$ in $T$ the values $\card{V_\tau}$ and $\chi_a(G_\tau)$
in accordance with \reflem{lem:acyStarJoinDisj}.

The second phase, in which an optimal acyclic and star coloring $\phi$ is constructed,
consists of a top-down traversal of the cotree.
Recall that, by \reflem{lem:acyStarJoinDisj},
for every \OneNode-node in $T$ with children $\alpha$ and $\beta$,
either $\alpha$ or $\beta$ must be saturated by $\phi$,
meaning that the vertices associated with leaves in the corresponding subtree
are assigned pairwise-distinct colors.
In order to decide which subtree is to be saturated,
the algorithm makes use of the values
$\card{V_\alpha}$, $\card{V_\beta}$, $\chi_a(G_\alpha)$, and $\chi_a(G_\beta)$,
which are all computed during the first phase.

\subsection{Phase I: computing \texorpdfstring{the quantity $\chi_a(G)=\chi_s(G)$}
                                              {the acyclic and star chromatic number}}
\algCOMPUTEAC
Recall that the initial phase of our main algorithm,
which is shown in \refalg{alg:phaseI},
computes the quantity $\chi_a(G_\tau) = \chi_s(G_\tau)$ for
every node $\tau$ in the cotree.
If one desires only the quantity $\chi_a(G) = \chi_s(G)$---but
not necessarily a corresponding coloring $\phi$ that achieves
this---then \refalg{alg:phaseI} is all that is required.
\begin{lem}\label{lem:phaseI}
  \refalg{alg:phaseI} is correct and runs in $O(n)$ time.
\end{lem}
\begin{proof}
  Correctness follows from an inductive application of \reflem{lem:acyStarJoinDisj}.
  Trivially, $\chi_a(G_\lambda) = 1$ for every leaf $\lambda$ of $T$.
  Observe that, when called on the root node $\rho$ of $T$,
  the procedure \COMPUTEAC{} correctly computes the quantity
  $\chi_a(G_\tau) = \chi_s(G_\tau)$ for every node $\tau$ in $T$.
  The correct value of the quantity $\chi_a(G) = \chi_s(G)$ is thus obtained
  as $\chi_a(G_\rho)$ when \COMPUTEAC{} is called on the root node $\rho$ of $T$.

  To establish the running time, observe that the recursive calls of \COMPUTEAC{} are
  in one-to-one correspondence with the nodes in $T$, of which there are $O(n)$.
\end{proof}

\subsection{Phase II: construction of an optimal coloring}
\algASSIGNCOLORS
The procedures for the second phase of our algorithm, which
require the values $\card{V_\tau}$ and $\chi_a(G_\tau)$ for every node $\tau$ in $T$
(as computed in the first phase),
are shown in \refalg{alg:phaseII}.
\begin{lem}\label{lem:phaseII}
  \refalg{alg:phaseII} is correct and runs in $O(n)$ time.
\end{lem}
\begin{proof}
  Let $\rho$ be the root of the given binary cotree $T$.
  We will show that the coloring $\phi$ obtained by calling \ASSIGNCOLORS{$\rho$, $1$}
  is an optimal acyclic (and star) coloring of $G$.

  Recall that two distinct vertices in $G$ are adjacent if and only if
  the lowest common ancestor of the corresponding leaves in $T$ is a \OneNode-node.
  It follows that any coloring that assigns disjoint sets of colors to the children
  of every \OneNode-node will be a proper coloring.
  It is easy to see that both \SATURATE{} and \ASSIGNCOLORS{} satisfy this condition,
  hence $\phi$ is a proper coloring of $G$,
  which means that we may apply \reflem{lem:acychar} to show that
  $\phi$ is an acyclic coloring of $G$.
  To this end, observe that the procedure \SATURATE{} behaves exactly as expected:
  the result of calling \SATURATE{$\tau$, $K$} is that the node $\tau$ is
  saturated by $\phi$ with the colors $\{K,\ldots,\card{V_\tau}-1\}$.
  Because \SATURATE{} is called on one of the two children of every \OneNode-node,
  we see that $\phi$ satisfies condition (ii) of \reflem{lem:acychar},
  from which it follows that $\phi$ is an acyclic coloring as desired.

  It remains to be shown that $\phi$ uses exactly $\chi_a(G)$ colors.
  Our proof of this fact is based upon an inductive argument concerning the
  recursive calls \ASSIGNCOLORS{$\tau$, $K$}.
  Trivially, when $\tau$ is a leaf, the induced subgraph $G_\tau$ is colored
  using $\chi_a(G_\tau) = 1$ colors.
  Now suppose $\tau$ is an internal node with children $\alpha$ and $\beta$.
  If $\tau$ is a \ZeroNode-node, and
  the calls \ASSIGNCOLORS{$\alpha$, $K$} and \ASSIGNCOLORS{$\beta$, $K$}
  result in optimal acyclic colorings of $G_\alpha$ and $G_\beta$,
  then, by \reflem{lem:acyStarJoinDisj},
  the induced subgraph $G_\tau$ is colored using
  $\chi_a(G_\tau) = \max\bigl\{\chi_a(G_\alpha),\chi_a(G_\beta)\bigr\}$ colors.
  Now suppose $\tau$ is a \OneNode-node and assume without loss of generality that
  $\card{V_\alpha}+\chi_a(G_\beta) \leq \card{V_\beta}+\chi_a(G_\alpha)$.
  If the call \ASSIGNCOLORS{$\beta$, $K+\card{V_\alpha}$}
  results in an optimal acyclic coloring of $G_\beta$,
  then, by \reflem{lem:acyStarJoinDisj},
  the induced subgraph $G_\tau$ is colored using
  $\chi_a(G_\tau) = \card{V_\alpha}+\chi_a(G_\beta)$ colors.

  As in the proof of \reflem{lem:phaseI}, we establish the running time via a
  one-to-one correspondence between recursive procedure calls and nodes in $T$.
  In this case, each node corresponds to either a call of \ASSIGNCOLORS{}
  or a call of \SATURATE{}.
\end{proof}

\section{Concluding Remarks}\label{sec:conc}

\refthm{thm:acyTW} implies a natural heuristic for the acyclic coloring problem:
simply find a triangulation $\Gt$ of $G$ that is close to optimal (with respect to treewidth),
and then compute an optimal proper coloring of $\Gt,$ using $O(n+m)$ time~\cite{Golumbic2004}.
Here we use the fact that treewidth is a particularly well-studied parameter,
and there are many heuristics, approximation algorithms,
exact (exponential) algorithms, and polynomial time algorithms for many classes of graphs~\cite{BodlaenderR03,Fomin2008,Kloks94}.
In particular, for a constant $k$ there is a linear-time algorithm for determining whether the treewidth
of a graph is at most $k$ and, if so, finding a corresponding triangulation~\cite{Bodlaender96}.

Furthermore, \reflem{lem:acyStarJoinDisj} applies to any graph that is decomposable with respect to the join operation,
and so it may be used as a reduction step that should be applied as the first step of any heuristic.
Moreover, \reflem{lem:acyStarJoinDisj} implies that we can also find an optimal acyclic
or star coloring of any graph for which these problems can be solved on all the graphs that result
from recursively applying the join decomposition.
For example, the \emph{tree-cographs}~\cite{Tinhofer1989}
are those graphs that result by taking disjoint unions and joins of trees or other tree-cographs.
The class of cographs is properly contained within this class.
Since it is trivial to find an optimal acyclic or star coloring of a tree in linear time~\cite{FRR2004},
it follows that we can solve these problems in linear time on the entire class of tree-cographs.

In the proof of \reflem{lem:acychar},
we were able to add at least one edge to every induced cycle on four vertices in $G$
(which was given along with an acyclic coloring)
such that no new induced cycles were created.
However, one can easily construct an example for general graphs where this is not the case.
Furthermore, there are graphs $G$ with acyclic colorings $\phi$ for which $G$ cannot be $\phi$-triangulated.
Two minimal examples are shown in \reffig{fig:coloredCantTriangulate}.
\begin{figure}[bht]
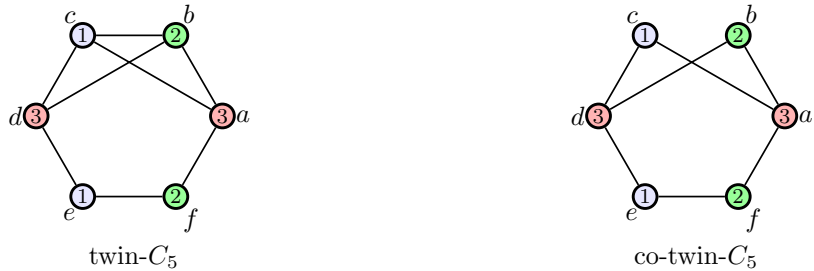
 \centering
  \begin{tabular}{ccc}
    \figCoTwinCfive{0.15\textwidth}{\normalsize}{\footnotesize} &
    \hspace{0.2\textwidth} &
    \figTwinCfive{0.15\textwidth}{\normalsize}{\footnotesize} \\
    twin-$C_5$ & & co-twin-$C_5$
  \end{tabular}
  \caption{Two graphs, each given with acyclic coloring $\phi$ such that neither can be $\phi$-triangulated.
           In the graph on the left, we cannot add an edge incident on $d$ or $a$ without creating a bichromatic cycle
           or violating the condition that the coloring is proper.
           Therefore, the cycles induced by $\{b,d,e,f,a\}$ and $\{c,d,e,f,a\}$ must be triangulated
           by adding edges $\{b,e\}$ and $\{c,f\},$ respectively.
           This results in $\{c,b,e,f\}$ inducing a bichromatic cycle.
           Note that edge $\{c,b\}$ must be added in any triangulation of the graph on the right,
           which reduces the problem to that of the graph on the left.}
  \label{fig:coloredCantTriangulate}
\end{figure}

In \reflem{lem:acychar} we proved the equivalence of the acyclic coloring and treewidth problems for
cographs by showing that every acyclic coloring of a cograph $G$ is a proper coloring of some triangulation of $G.$
It would be useful to prove similar results for other classes of graphs;
it is natural to consider other classes for which the treewidth problem can be solved in polynomial time.

\section*{Acknowledgments}
The author wishes to thank the anonymous referees, whose
suggestions greatly improved the presentation of this paper.
The majority of this work was conducted when the author was employed by
The University of Chicago and Argonne National Laboratory,
supported in part by U.S. Department of Energy under Contract DE-AC02-06CH11357.

\bibliographystyle{abuser} 
\bibliography{ColoringCographs}

\end{document}